\documentclass[a4paper,12pt]{article}

\usepackage{a4wide,amsthm,amsfonts,amsmath} 

\title{On polynomial growth functions of D0L-systems}
\author{Fran\c{c}ois Nicolas \and Julien Cassaigne}

\newcommand{\N}{\mathbb{N}} 
\newcommand{\R}{\mathbb{R}}
\newcommand{\lgr}[1]{\left| #1 \right|}
\newcommand{\seg}[2]{\left[  #1, #2 \right]}
\newcommand{\defeq}{\mathrel{\mathop:}=}

\newcommand{\mv}{\varepsilon}
 \newcommand{\ttx}{\mathtt{x}}
\newcommand{\ze}{\mathtt{0}}
\newcommand{\on}{\mathtt{1}}
\newcommand{\tw}{\mathtt{2}}
\newcommand{\thr}{\mathtt{3}}
\newcommand{\tta}{\mathtt{a}}
\newcommand{\ttb}{\mathtt{b}}
\newcommand{\tte}{\mathtt{e}}
\newcommand{\calf}{\mathcal{F}}

\newtheorem{lemma}{Lemma}
\newtheorem{theorem}{Theorem}
\newtheorem{definition}{Definition}
\newtheorem{exercise}{Exercise}

\begin{document}

\maketitle

\sloppy 

\begin{abstract}
The aim of this paper is to prove that every polynomial function that maps the natural integers to the positive integers is the growth function of some D0L-system.
\end{abstract}

\section{Introduction}

As usual,
 $\N$ and $\R$ denote the semiring of natural integers and the field of real numbers, respectively.
For every $a$, $b \in \N$, 
$\seg{a}{b}$ denotes the set of all $n\in \N$ such that $a \le n \le b$.
The letter $\ttx$ denotes the formal variable for the polynomials.

A \emph{word} is a finite string of symbols. 
Word concatenation is denoted multiplicatively.
For every word $w$, the \emph{length} of $w$ is denoted  $\lgr{w}$. 
The word of length zero, denoted $\mv$,  is called the \emph{empty word}.

An \emph{alphabet} is a finite set of symbols.
Let $A$ be an alphabet. 
The set of all words over $A$ is denoted $A^\star$.
A mapping $\sigma\colon A^\star \to A^\star$ is called a \emph{morphism} if $\sigma(xy) = \sigma(x) \sigma(y)$ for all $x$, $y \in A^\star$.
Clearly, 
$\sigma$ is completely determined  by its restriction to $A$.
For every $n \in \N$, $\sigma^n$ denotes the $n^\text{th}$ iterate of $\sigma$: 
for every $w \in A^\star$, 
$\sigma^0(w) = w$, 
$\sigma^1(w) = \sigma(w)$,
$\sigma^2(w) = \sigma(\sigma(w))$, 
$\sigma^3(w) = \sigma(\sigma(\sigma(w)))$, \emph{etc}.

A \emph{D0L-system}  is a triple $(A, \sigma, w)$, 
 where $A$ is an alphabet, $\sigma$ is a morphism from $A^\star$ to itself, 
and $w$ is a word over $A$.  
The \emph{growth function} of the D0L-system $(A, \sigma, w)$ is defined as the function mapping each $n \in \N$ to the length of $\sigma^n(w)$. 
By extension, we say that  a (formal) real polynomial $F(\ttx)$ is the growth function of $(A, \sigma, w)$ if  $\lgr{\sigma^n(w)} = F(n)$ for every $n \in \N$.
We say that a function $f\colon \N \to \R$ (resp.~a real polynomial $F(\ttx)$)  is a \emph{D0L-growth function} if $f$ (resp.~$F(\ttx)$) is the growth function of some D0L-system.

\begin{definition}
Define  $\calf$ as the set of all polynomials $F(\ttx)$ such that $F(n)$ is a positive integer for every $n \in \N$. 
\end{definition}

Every polynomial with non-negative integer coefficients and non-zero constant term belongs to $\calf$:  
for instance, $\ttx^d + 1 \in \calf$ for every $d \in \N$.
Furthermore,
 some polynomials that admit negative and/or non-integer coefficients also belong to $\calf$: 
for each integer $d \ge 1$, 
$\ttx^d - \ttx^{d - 1} + 1 \in \calf$ 
and
$\frac{1}{d!} \prod_{k = 1}^d (\ttx + k) \in \calf$.

\begin{exercise}
 Prove that every polynomial in $\calf$ has rational coefficients (hint: use Lagrange polynomials).
\end{exercise}

\section{Result}

The aim of the paper is to prove and exemplify the following result:

\begin{theorem} \label{th:D0L-polynomial}
Every polynomial in $\calf$ is a D0L-growth function.
 \end{theorem}

Clearly, every D0L-growth function is non-negative valued.
Moreover, if a D0L-growth function vanishes then it is eventually zero because every morphism maps the empty word to itself.
Therefore, if a non-zero polynomial is a D0L-growth function then it belongs to $\calf$: Theorem~\ref{th:D0L-polynomial} is optimal.

\begin{definition} \label{def:ev-pos}
We say that a real polynomial $F(\ttx)$ is \emph{eventually positive} if it satisfies the following four equivalent conditions.
\begin{enumerate}
	\item The set $\left\{ t \in \R : F(t) \le 0 \right\}$ is bounded from above.
	\item The set $\left\{ t \in \R : F(t) > 0 \right\}$  is not bounded from above.
	\item $F(\ttx)$ is non-zero and its leading coefficient is positive.
	\item Either $F(\ttx)$ is a positive constant or  $\lim\limits_{t \to + \infty} F(t) = + \infty$.
\end{enumerate}
\end{definition}

\begin{exercise}
 Check that conditions~1, 2, 3 and 4 of Definition~\ref{def:ev-pos} are equivalent.
\end{exercise}

Condition~1 is the natural definition of the notion: 
given an upper-bound $t_0$ of $\left\{ t \in \R : F(t) \le 0 \right\}$,
$F(t)$ is positive for every real number  $t > t_0$.
Condition~2 ensures that every polynomial in $\calf$ is eventually positive: for every  $F(\ttx) \in \calf$, $\N$ is a subset of $\left\{ t \in \R : F(t) > 0 \right\}$.
Condition~3 is used in the proof of Lemma~\ref{lem:Fx-1-Fx} below.
Condition~4 is given for the sake of completeness.

\begin{definition}
For each polynomial $F(\ttx)$, define $\partial F(\ttx) \defeq F(\ttx + 1) - F(\ttx)$.
For each function $f\colon \N \to \R$, 
define $\partial f\colon \N \to \R$ by:
$\partial f(n) \defeq f(n + 1) - f(n)$ for each $n \in \N$.
\end{definition}

Using a telescoping sum we obtain that
\begin{equation} \label{eq:f-n-g}
f(n) = f(0) + \sum_{m = 0}^{n - 1} \partial f(m)
\end{equation}
for any function $f\colon \N \to \R$ and any $n \in \N$.

\begin{lemma} \label{lem:Fx-1-Fx}
Let $F(\ttx)$ be a non-constant real polynomial.

\begin{enumerate}
\item The degree of $\partial F(\ttx)$ is one less than the degree of $F(\ttx)$.
\item $F(\ttx)$ is eventually positive if, and only if, $\partial F(\ttx)$ is eventually positive.
\end{enumerate}
\end{lemma}

\begin{proof}
Let $d$ denote the degree of $F(\ttx)$ and let $f_0$, $f_1$, $f_2$, \ldots, $f_d \in \R$ be such that:
$$
F(\ttx) = \sum_{k = 0}^d f_k \ttx^k \, .
$$
Clearly, $\partial F(\ttx)$ can be written in the form 
$$
\partial F(\ttx) =  \sum_{k = 1}^d  \left( {(\ttx + 1)}^k - \ttx^k  \right) f_k \, .
$$
Now,
remark that for each integer $k \ge 1$, 
the polynomial ${(\ttx + 1)}^k - \ttx^k$ is of degree $k - 1$.
Hence,  $\partial F(\ttx)$ is of degree $d - 1$.
Furthermore, the leading coefficient of ${(\ttx + 1)}^d - \ttx^d$ equals $d$, and thus the leading coefficient of  $\partial F(\ttx)$ equals $d f_d$: 
$d$ times the leading coefficient of $F(\ttx)$. 
\end{proof}

\begin{exercise} \label{exo:integer-valued}
A polynomial $F(\ttx)$ is called \emph{integer-valued} if $F(n)$ is a rational integer for every rational integer $n$.
Let $F(\ttx)$ be a non-zero real polynomial and let $d$ denote the degree of $F(\ttx)$.
Prove that $F(\ttx)$ is integer-valued if, and only if, $F(n)$ is a rational integer for every $n \in \seg{0}{d}$.
\end{exercise}

It follows from Exercise~\ref{exo:integer-valued} that every polynomial in $\calf$ is integer-valued: 
for every $F(\ttx) \in \calf$ and every $n \in \N$, $F(- n)$ is an integer. 

\begin{lemma} \label{lem:shift-D0L}
Let $f$, $g\colon \N \to \R$ be such that $f(0)$ is a positive integer and $g(n) = f(n + 1)$ for every $n \in \N$.
Then, 
 $g$  is a D0L-growth function if, and only if, $f$ is a D0L-growth function.
\end{lemma} 

\begin{proof}
If $f$ is the  growth function of some D0L-system $(A, \sigma, x)$ then $g$ is the growth function of the D0L-system $(A, \sigma, \sigma(x))$.

Conversely, assume that $g$ is the growth function of some D0L-system $(B, \tau, y)$.
Let $a$ and $c$  be two letters such that $a \ne c$, $a \notin B$ and $c \notin B$.
Let $A \defeq B \cup \{ a, c \}$, 
let $x \defeq  c^{f(0) - 1} a$, and 
let $\sigma\colon A^\star \to A^\star$ be the morphism defined by:
$\sigma(a) \defeq  y$, 
$\sigma(b) \defeq \tau(b)$ for every $b \in B$, and 
$\sigma(c) \defeq \mv$.
It is easy to see that $f$ is the growth function of the D0L-system $(A, \sigma, x)$:
$\sigma^n(x) = \sigma^{n - 1}(y) = \tau^{n - 1}(y)$ for every integer $n \ge 1$.
\end{proof}

\begin{lemma} \label{lem:partialf}
Let $f\colon \N \to \R$ be such that $f(0)$ is a positive integer.
If $\partial f$ is a D0L-growth function then $f$ is also a D0L-growth function.
\end{lemma}

\begin{proof}
Assume that $\partial f$ is the growth function of some D0L-system $(B, \tau, y)$. 
Let $a$ and $c$  be two letters such that $a \ne c$, $a \notin B$ and $c \notin B$.
Let $A \defeq B \cup \{ a, c \}$, let $x \defeq  c^{f(0) - 1} a$ and let $\sigma\colon A^\star \to A^\star$ be the morphism defined by:
$\sigma(a) \defeq  a y$, 
$\sigma(b) \defeq \tau(b)$ for every $b \in B$, and 
$\sigma(c) \defeq c$.
It is easy to see that $\sigma^n(a) = a y \tau(y) \tau^2(y) \tau^3(y) \dotsm \tau^{n - 1}(y)$ and that  
\begin{equation} \label{eq:sigma-n-a}
\sigma^n(x) = c^{f(0) - 1} a y \tau(y) \tau^2(y) \tau^3(y) \dotsm \tau^{n - 1}(y) 
\end{equation}
for every $n \in \N$.
Now, remark that the right-hand side of Equation~\eqref{eq:f-n-g} is exactly the length of the word on the right-hand side of Equation~\eqref{eq:sigma-n-a}: 
$f$ is the growth function of the D0L-system $(A, \sigma, x)$.
\end{proof}

We can now prove the main result of the paper.

\begin{proof}[Proof of Theorem~\ref{th:D0L-polynomial}]
We proceed by induction on the degree of $F(\ttx)$. 

If $F(\ttx)$ has degree zero then $F(\ttx)$ is identically equal to $F(0)$, 
and thus $F(\ttx)$ is the growth function of the D0L-system $(A, \sigma, x)$ where 
$A \defeq \{ \tta \}$, 
$\sigma$ is the identity function on ${\{ \tta \}}^\star$, and 
$x \defeq \tta^{F(0)}$.

Let us now assume that the degree of $F(\ttx)$, denoted $d$, is positive.
Clearly,  $\partial F(n)$ is an integer for every $n \in \N$.
Moreover, Lemma~\ref{lem:Fx-1-Fx} ensures that  $\partial F(\ttx)$ is of degree $d - 1$ and  eventually positive.
Therefore,
there exists  $n_0 \in \N$ such that $\partial F(\ttx + n_0) \in \calf$, 
and it follows from the induction hypothesis that $\partial F(\ttx + n_0)$ is a D0L-growth function. 
Then, Lemma~\ref{lem:partialf} ensures that $F(\ttx + n_0)$ is also a D0L-growth function.
Repeatedly applying Lemma~\ref{lem:shift-D0L}, we obtain that 
$F(\ttx + k)$ is a D0L-growth function for $k = n_0 - 1$, $n_0 - 2$, $n_0 - 3$, \ldots, $1$, $0$.
In particular, $F(\ttx)$ is a D0L-growth function.
\end{proof}

\section{Effectivity and examples}

Let $F(\ttx) \in \calf$. 
Summarizing the previous section, we present a simple algorithm that computes a D0L-system with growth function  $F(\ttx)$.

The degree of $F(\ttx)$ is denoted $d$.
For each $i \in \N$, 
$\partial^i$ is understood as the $i^\text{th}$ iterate of the operator $\partial$: 
for every polynomial $F(\ttx)$, 
$\partial^0 F(\ttx) = F(\ttx)$, 
$\partial^1 F(\ttx) = \partial F(\ttx)$,
$\partial^2 F(\ttx)  = F(\ttx + 2) -  2 F(\ttx + 1) + F(\ttx)$,
\emph{etc}.

\begin{exercise} \label{exo:partial-k-Fx}
Prove that 
$$
\partial^i f(n) = \sum_{j = 0}^i \binom{i}{j} {(-1)}^{i - j} f(n + j) 
$$
for every $i$, $n \in \N$ and every function $f\colon \N \to \R$.
\end{exercise}

It is easy to see that $\partial^i F(n)$ is a rational integer for all $i$, $n \in \N$.

\begin{exercise} \label{exo:partial-zero}
Let $F(\ttx) \in \calf$ and let $d$ denote the degree of $F(\ttx)$.
Prove that if  $\partial^i F(0)$ is positive for every $i \in \seg{0}{d}$ then 
$\partial^i F(\ttx)$ belongs to $\calf$ for every $i \in \seg{0}{d}$.
\end{exercise}

Exercise~\ref{exo:partial-zero} motivates the introduction of the following restricted case of the problem.

\subsection{Restricted case}

For each $i \in \N$,  let $f_i \defeq \partial^i F(0)$.
In this section, we focus on the case where $f_i$ is a positive integer for every $i \in \seg{0}{d}$.
It follows from Lemma~\ref{lem:Fx-1-Fx} that the polynomial $\partial^i F(\ttx)$ has degree $d - i$ for each $i \in \seg{0}{d}$.
In particular, $\partial^d F(\ttx)$ is identically equal to the positive integer $f_d$, and thus $\partial^d F(\ttx)$ is a D0L-growth function.
Then, applying $d$ times Lemma~\ref{lem:partialf}, 
we get that the polynomials 
$\partial^{d - 1} F(\ttx)$,  \ldots, $\partial^{3} F(\ttx)$,  $\partial^{2} F(\ttx)$, $\partial F(\ttx)$, and $F(\ttx)$ are D0L-growth functions.
Hence, to obtain a D0L-system with growth function $F(\ttx)$, it suffices to nest $d$ instances of the gadget described in the proof of Lemma~\ref{lem:partialf}. 
Let $a_0$, $a_1$, $a_2$, \ldots, $a_d$ be $d + 1$ pairwise distinct letters and let $A \defeq \{ a_0, a_1, a_2, \dotsc, a_d \}$. 
For each $i \in \seg{0}{d}$, put $x_i \defeq  a_0^{f_{d - i} - 1} a_i$.
Let $\sigma\colon A^\star \to A^\star$ be the morphism defined by:
$\sigma(a_0) \defeq a_0$ and $\sigma(a_i) \defeq a_i x_{i - 1}$ for each $i \in \seg{1}{d}$.
The polynomial $F(\ttx)$ is the growth function of the D0L-system $(A, \sigma, x_d)$.
More generally, $\partial^{i} F(\ttx)$  is the growth function of the D0L-system $(A, \sigma, x_{d - i})$ for every $i \in \seg{0}{d}$.

For instance, let $F(\ttx) \defeq  \ttx^3 + 1$: 
\begin{itemize} 
 \item 
$d = 3$,  $A =  \{ \ze, \on, \tw, \thr  \}$,
\item 
$\partial F(\ttx) = 3 \ttx^2 + 3 \ttx + 1$,
$\partial  \partial  F(\ttx) = 6 \ttx + 6$,
$\partial \partial  \partial  F(\ttx) = 6$,
\item 
$x_0 = \ze\ze \ze \ze \ze \ze$, 
$x_1 = \ze \ze \ze \ze \ze \on$,
$x_2 = \tw$, 
$x_3 = \thr$,
\item 
$\sigma(\ze) = \ze$,  
$\sigma(\on) = \on \ze\ze\ze\ze\ze\ze$,
$\sigma(\tw) = \tw  \ze\ze\ze\ze\ze \on$, 
$\sigma(\thr) = \thr \tw $, 
\item
$\lgr{\sigma^n(x_0)} = 6$,
$\lgr{\sigma^n(x_1)} = 6n + 6$,
$\lgr{\sigma^n(x_2)} = 3 n^2 + 3n + 1$,
$\lgr{\sigma^n(x_3)} = n^3 + 1$.
\end{itemize}

 \subsection{General case}

In this section, we show how to reduce the general case to the restricted case.
It follows from Lemma~\ref{lem:Fx-1-Fx} that the polynomial $\partial^iF(\ttx)$ is eventually positive for every $i \in \seg{0}{d}$.
Hence, there exists $k \in \N$ such that $\partial^i F(k)$ is positive for every $i \in \seg{0}{d}$.
Note that $k$ is computable by exhaustive search.

\begin{exercise}
Find an algorithm that, for any rational polynomial $F(\ttx)$ taken as input, decides whether $F(\ttx)$ belongs to $\calf$ (hint: use Exercises~\ref{exo:integer-valued} and~\ref{exo:partial-zero}). 
\end{exercise}

 \newcommand{\tildeF}{G} 

The method for the restricted case applies to $\tildeF(\ttx) \defeq F(\ttx + k)$:
we can compute a D0L-system $(B, \tau, y)$ such that $\lgr{\tau^n(y)} = \tildeF(n)$ for every $n \in \N$.
To obtain a D0L-system with growth function $F(\ttx)$, 
it suffices to nest $k$ instances of the gadget described in the proof of Lemma~\ref{lem:shift-D0L}.
Let $e$, $b_1$, $b_2$,  $b_3$, \ldots, $b_{k}$ be $k + 1$   pairwise distinct letters such that $\{ e, b_1, b_2,  b_3, \ldots, b_{k} \} \cap B = \emptyset$,
let $A \defeq \{ e, b_1, b_2,  b_3, \ldots, b_{k} \} \cup B$, and let $\sigma\colon A^\star \to A^\star$ be the morphism defined by:
$\sigma(e) \defeq \mv$, 
$\sigma(b_i) \defeq e^{F(i) - 1} b_{i + 1}$ for each $i \in \seg{1}{k - 1}$,
$\sigma(b_k) \defeq y$,
and 
$\sigma(b) \defeq \tau(b)$ for every $b \in B$.
The polynomial $F(\ttx)$ is the growth function of the D0L-system $(A, \sigma, \tte^{F(0) - 1} b_1)$.

For instance, let $F(\ttx) \defeq {(\ttx - 2)}^2 + 2$.
Remark that $\partial F(0) = - 3$ and $\partial F(1) = - 1$.
However, the polynomial $\tildeF(\ttx) \defeq F(\ttx + 2)$ is such that $\partial^i \tildeF(0) > 0$  for each $i \in \{ 0, 1, 2 \}$:
$\tildeF(\ttx) = \ttx^2 + 2$,
$\partial \tildeF(\ttx) = 2 \ttx + 1$, and 
$\partial  \partial \tildeF(\ttx) = 2$.
Let $A \defeq \{  \tte, \tta, \ttb,  \ze, \on, \tw  \}$ and let $\sigma\colon  {A}^\star \to {A}^\star$ be the morphism defined by:
\begin{itemize}
	\item 
$\sigma(\tte) \defeq \mv$, $\sigma(\tta) \defeq  \tte \tte \ttb$, $\sigma(\ttb) \defeq \ze \tw$,
	\item 
$\sigma(\ze) \defeq \ze$,  
$\sigma(\on) \defeq \on \ze \ze$,
$\sigma(\tw) \defeq \tw \on$.
\end{itemize}
It is easy to see that
$\lgr{\sigma^n(\ze \ze)} = 2$,
$\lgr{\sigma^n(\on)} = 2n + 1$, 
$\lgr{\sigma^n(\ze \tw  )} = n^2 + 2$ 
$\lgr{\sigma^n(\tte \tte \ttb)} = {(n - 1)}^2 + 2$, 
and 
$\lgr{\sigma^n(\tte \tte \tte \tta)} = {(n - 2)}^2 + 2$  for every $n \in \N$.
Hence, the polynomial  $F(\ttx)$  is  the growth function of the D0L-system $(A, \sigma, \tte \tte \tte \tta)$.

\end{document}